\documentclass[12pt, journal, draft, onecolumn]{IEEEtran} 

\usepackage[cmex10]{amsmath}
\usepackage{amssymb}
\usepackage{mathrsfs}
\usepackage[final]{graphicx}

\newtheorem{remark}{Remark}
\newcommand{\expectation}{\ensuremath{\mathbb{E}}}
\newcommand{\Expt}{\expectation}
\usepackage{url}

\begin{document}

\title{Extremality Properties for the Basic Polarization Transformations}
\author{\IEEEauthorblockN{Mine Alsan}\\
\small\IEEEauthorblockA{Information Theory Laboratory\\
Ecole Polytechnique F\' ed\' erale de Lausanne\\
CH-1015 Lausanne, Switzerland\\
Email: mine.alsan@epfl.ch}
\normalsize}
\maketitle
\pagestyle{empty}
\thispagestyle{empty}
\IEEEpeerreviewmaketitle

\let\thefootnote\relax\footnotetext{The material in this paper was presented in part at the IEEE International Symposium on Information Theory, Boston, USA,
July 2012.}

\begin{abstract}
We study the extremality of the BEC and the BSC for Gallager's reliability function
$E_0$ evaluated under the uniform input distribution for binary input DMCs from the aspect of channel polarization. 
In particular, we show that amongst all B-DMCs of a given $E_0(\rho)$ value, for a fixed $\rho\geq 0$, 
the BEC and BSC are extremal in the evolution of $E_0$ under the one-step polarization transformations. 

\end{abstract}


\section{Introduction}
While the capacity of a memoryless channel $W$ gives the largest
rate that may be communicated reliably across it, the reliability function
$E(R,W)$ provides a finer measure on the quality of the channel: for any
rate $R$ less than channel capacity, it is possible to find a sequence of
codes of increasing blocklength, each of which of rate at least $R$, and
whose block error probability decays exponentially to zero in the
blocklength ---  $E(R,W)$ is the largest possible rate of this decay.

Gallager classical treatise~\cite{578869} gives a lower bound to $E(R,W)$,
the random coding exponent $E_r(R,W)$ in the form $E_r(R,W)=\max_{\rho\in[0,1]}
E_0(\rho,W)-\rho R$.  Remarkably, this lower bound is tight for rates above
the critical rate $E_0'(1,W)$.  The function $E_0(\rho,W)$ that appears
as an auxiliary function on the road to deriving $E_r(R,W)$ turns out to
be of independent interest in its own right.  In particular, $E_0(\rho,W)/\rho$
is the largest rate for which a sequential decoder can operate while keeping
the $\rho$-th moment of the decoder's computation effort per symbol bounded.

In \cite{6284065}, we investigated the extremal properties of $E_0(\rho,W)$ evaluated under the uniform input distribution 
for the class of binary input channels.  We have shown that among all
such channels with a given value of $E_0(\rho_1,W)$, for $\rho_1\in[0, 1]$, the binary erasure channel
(BEC) and the binary symmetric channel (BSC) distinguish themselves in
certain ways: they have, respectively, the largest and smallest value of
$E_0'(\rho_2,W)$ for any $\rho_2\in[\rho_1, 1]$. Furthermore, we showed that amongst channels $W$ with a given value of
$E_0(\rho, W)$ for a given $\rho\in[0, 1]$, the BEC and BSC are the most and least
polarizing under Ar{\i}kan's polar transformations in the sense that their
polar transforms $W^+$ and $W^-$ have the largest and smallest difference
in their $E_0$ values.

In this paper, we extend the result related to the BEC and BSC being extremal for Ar{\i}kan's
polarization transforms to the region where $\rho \geq 0$. In his award winning paper \cite{1669570}, Ar{\i}kan describes two
synthetic channels $W^+$, and $W^-$ which can be obtained from two independent
copies of $W$.  It is well known (proved as a corollary to extremes of
information combining) that among all channels $W$ with a given
symmetric capacity $I(W)$, the BEC and BSC polarize most and least in
the sense of having the largest and smallest difference between $I(W^+)$
and $I(W^-)$.  We report a more general conclusion: amongst all channels
$W$ with a given value of $E_0(\rho,W)$, the BEC and BSC polarize most and
least in the sense of having the largest difference between $E_0(\rho,W^+)$
and $E_0(\rho,W^-)$ whenever $\rho\in[0, 1]\cup[2, \infty]$. 
On the other hand, for $\forall \rho\in[1, 2]$, we show that the BEC maximizes, and the BSC minimizes the $E_0$ values obtained after both applying the $W^+$, or the $W^-$ transformations.  

\subsection{Definitions}
Given a binary input channel $W$, let $E_0(\rho,W)$ denote ``Gallager's $E_0$'' \cite[p.~138]{578869} evaluated for the uniform input distribution:
\begin{equation} \label{eq:E0}
E_0(\rho, W) = -\log \displaystyle\sum_{y\in \mathcal{Y}} \left[ \frac{1}{2} W(y\mid 0)^{\frac{1}{1+\rho}} + \frac{1}{2} W(y\mid 1)^{\frac{1}{1+\rho}} \right]^{1+\rho}.   
\end{equation} 

Theorem 5.6.3 in~\cite{578869} summarizes the properties of $E_0(\rho, W)$ with respect to the variable $\rho$. For $\rho \geq 0$, $E_0(\rho, W)$ 
is a positive, concave increasing function in $\rho$. Moreover, the symmetric capacity $I(W)$ of the channel
can be derived from $E_0(\rho, W)$ by  
\begin{equation}\label{eq::I_def} 
\lim_{\rho\to 0}\frac{\displaystyle E_0(\rho,W)}{\displaystyle \rho} = \frac{\displaystyle \partial}{ \displaystyle \partial \rho} E_0(\rho, W)\Bigl\lvert_{\rho=0} = I(W)
\end{equation}
and the Bhattacharyya parameter $Z(W)$ from the cut-off rate as
\begin{equation}\label{eq::Z_def} 
E_0(1, W) = \log \frac{\displaystyle 2}{ \displaystyle 1 + Z(W)}.
\end{equation}

The next lemma due to Telatar and Ar{\i}kan~\cite{notes1} introduces a useful representation for the $E_0(\rho, W)$ parameter.
\newtheorem{lemma}{Lemma}
\begin{lemma}\label{lem:basic}\cite{notes1}
Given a symmetric B-DMC $W$, and a fixed $\rho \in [0, 1]$, there exist a random variable $Z$ taking values in the $[0, 1]$ interval such that
\begin{equation}\label{eq:Eo}
 E_0( \rho, W) = -\log{\Expt\left[g(\rho, Z)\right]}   
\end{equation}
where
\begin{equation}\label{eq:g}
 g(\rho, z) =  \left(\frac{1}{2}\left( 1 + z\right)^{\frac{1}{1+\rho}} + \frac{1}{2}\left(1 - z\right)^{\frac{1}{1+\rho}} \right)^{1+\rho}.
\end{equation}
Moreover, the random variable $Z_{\mathsf{BEC}}$ of a binary erasure channel is $\{0, 1\}$ valued. The random variable $Z_{\mathsf{BSC}}$ of a
binary symmetric channel is a constant $z_{\mathsf{BSC}}$. 
\end{lemma}
\begin{proof}
Recall
$\displaystyle
   E_0( \rho, W) = -\log\displaystyle\sum_{y} \bigg[ \hspace{2mm} \frac{1}{2} W(y \mid 0)^{\frac{1}{1+\rho}} + \frac{1}{2} W(y \mid 1)^{\frac{1}{1+\rho}} \hspace{2mm} \bigg]^{1+\rho} \nonumber 
$.
Define
\begin{equation}\label{eq:dist}
 q_{W}(y)=\frac{W(y\mid0) + W(y\mid1)}{2} \quad\text{and}\quad \Delta_{W}(y) = \frac{W(y\mid0)-W(y\mid1)}{W(y\mid0)+W(y\mid1)} \\
\end{equation}
so that
 $W(y\mid0)= q_{W}(y)[1+\Delta_{W}(y)]$ and $W(y\mid1) = q_{W}(y)[1-\Delta_{W}(y)]$.
Then, one can define the random variable $Z =  \lvert\Delta_{W}(Y)\rvert\in [0, 1]$ where $Y$ has the probability distribution $q_{W}(y)$, and obtain 
\eqref{eq:Eo} by simple manipulations. The claims about $Z_{\mathsf{BEC}}$, and $Z_{\mathsf{BSC}}$ are verified easily from \eqref{eq:dist}.
\end{proof}

\section{Extremality Results for the Polarization Transformations}

\subsection{Basic Polarization Transformations}
In~\cite{1669570}, a low complexity code construction that achieves the symmetric capacity of B-DMCs is given based on the recursive application of two basic channel transformations.
These transforms, usually refered as the minus and plus transformations, synthesize two new channels by combining two independent copies of a given channel. The transition probabilities of the new channels
are defined in terms of the initial one by the definitions given in \cite[Eqs.~(19), (20)]{1669570}.

Instead of identical copies of a given channel, we propose to combine two independent copies of different B-DMCs in a similar way. We denote by
$W_{1,2}^{-}: \mathcal{X} \rightarrow \mathcal{Y}^{2}$ and $W_{1,2}^{+}: \mathcal{X} \rightarrow \mathcal{Y}^{2}\times\mathcal{X}$ the synthesized channels obtained by combining independent copies of the channels $W_1$ and $W_2$. 
In this case, the transition probabilities can be defined by
\begin{align} 
 \label{align:trans1} &W_{1,2}^{-}(y_1 y_2 \mid u_1) = \displaystyle\sum_{u_2\in\mathcal{X}} \frac{1}{2} W_1(y_1 \mid u_1 \oplus u_2) W_2(y_2 \mid u_2) \\
 \label{align:trans2} &W_{1,2}^{+}(y_1 y_2 u_1\mid u_2) = \frac{1}{2} W_1(y_1 \mid u_1 \oplus u_2) W_2(y_2 \mid u_2). 
\end{align}
The following two lemmas express the $E_0$ parameter of the synthesized channels $W_{1,2}^{-}$, and $W_{1,2}^{+}$ in terms of the representation given in Lemma \ref{lem:basic}, relating them to the $E_0$ parameters of the channels $W_1$ and $W_2$.

\begin{lemma}\label{lem:minus}
Given two B-DMCs $W_1$, $W_2$, and $\rho\geq0$, let $Z_1$ and $Z_2$ be independent RVs such that
\begin{equation*}
  E_0( \rho, W_1) = -\log{\Expt\left[ g(\rho, Z_1)\right] }   \quad \text{and} \quad E_0( \rho, W_2) = -\log{\Expt\left[ g(\rho, Z_2)\right] } 
\end{equation*}
hold as defined in Lemma~\ref{lem:basic}. Then,
\begin{equation}\label{eq:Eominus}
 E_0( \rho, W_{1,2}^{-}) = -\log{\Expt\left[ g(\rho, Z_1 Z_2)\right] }   
\end{equation}
where $g(\rho, z)$ is given by \eqref{eq:g}. 
\end{lemma}
\begin{proof}
From the definition of the channel $W_{1,2}^{-}$ in \eqref{align:trans1}, we can write
\begin{align}
 E_0( \rho, W_{1,2}^{-}) = -\log\displaystyle\sum_{y_1, y_2} &\bigg[ \hspace{2mm} \frac{1}{2} W_{1,2}^{-}(y_1, y_2 \mid 0)^{\frac{1}{1+\rho}} + \frac{1}{2} W_{1,2}^{-}(y_1, y_2 \mid 1)^{\frac{1}{1+\rho}} \hspace{2mm} \bigg]^{1+\rho} \nonumber \\
= -\log\displaystyle\sum_{y_1, y_2} &\bigg[ \hspace{2mm} \frac{1}{2} \left( \frac{1}{2} W_1(y_1 \mid 0) W_2(y_2 \mid 0) +  \frac{1}{2} W_1(y_1 \mid 1) W_2(y_2 \mid 1) \right)^{\frac{1}{1+\rho}} \nonumber \\
&+ \frac{1}{2} \left( \frac{1}{2} W_1(y_1 \mid 1) W_2(y_2 \mid 0) +  \frac{1}{2} W_1(y_1 \mid 0) W_2(y_2 \mid 1) \right)^{\frac{1}{1+\rho}} \bigg]^{1+\rho} \nonumber \\
= -\log \displaystyle\sum_{y_1 y_2} &\left[ \hspace{2mm} \frac{1}{2} \left(\frac{1}{2}\right)^{\frac{1}{1+\rho}} q_{W_1}\displaystyle \left(y_1\right)^{\frac{1}{1+\rho}} q_{W_2}\left(y_2\right)^{\frac{1}{1+\rho}} \right. \nonumber \\
& \hspace{4mm} \bigl( \left( 1 + \Delta_{W_1}\left(y_1\right)\right)  \left( 1 + \Delta_{W_2}\left(y_2\right)\right)  + \left( 1 - \Delta_{W_1}\left(y_1\right)\right)  \left( 1 - \Delta_{W_2}\left(y_2\right)\right) \bigr) ^{\frac{1}{1+\rho}} \nonumber \\
& + \displaystyle  \bigl( \left( 1 - \Delta_{W_1}\left(y_1\right)\right)  \left( 1 + \Delta_{W_2}\left(y_2\right)\right) + \left( 1 + \Delta_{W_1}\left(y_1\right)\right)  \left( 1 - \Delta_{W_2}\left(y_2\right)\right)\bigr)^{\frac{1}{1+\rho}} \hspace{2mm} \bigg]^{1+\rho} \nonumber  \\
= -\log \displaystyle\sum_{y_1 y_2} &\hspace{2mm} q(y_1) \hspace{2mm} q(y_2) \hspace{2mm} \left[\frac{1}{2} \bigl(1 + \Delta_{W_1}(y_1) \Delta_{W_2}(y_2)\bigr)^{\frac{1}{1+\rho}} + \frac{1}{2} \bigl(1-\Delta_{W_1}(y_1)\Delta_{W_2}(y_2)\bigr)^{\frac{1}{1+\rho}}\right]^{1+\rho} \nonumber
\end{align}
where we used the definitions in \eqref{eq:dist}. We can now define $Z_1 = \lvert\Delta_{W_1}(Y_1)\rvert$ and $Z_2 = \lvert\Delta_{W_2}(Y_2)\rvert$ 
where $Y_1$ and $Y_2$ are independent random variables with distribution $q_{W_1}$ and $q_{W_2}$, respectively.
From this construction, the lemma follows. 
\end{proof}

\begin{lemma}\label{lem:plus}
Given two B-DMCs $W_1$, $W_2$, and $\rho\geq0$, let $Z_1$ and $Z_2$ be as in Lemma~\ref{lem:minus}.  Then,
\begin{equation}\label{eq:Eoplus}
E_0( \rho, W_{1,2}^{+}) = -\log{\Expt{\left[ \hspace{2mm} \frac{1}{2} \bigl(1 + Z_1 Z_2\bigr)\hspace{2mm} g\Bigl( \rho, \frac{Z_1 + Z_2}{1 + Z_1 Z_2}\Bigr) +\frac{1}{2} \bigl(1 - Z_1 Z_2\bigr)\hspace{2mm} g\Bigl( \rho, \frac{Z_1 - Z_2}{1 - Z_1 Z_2}\Bigr) \hspace{2mm} \right]}} 
\end{equation}
where $g(\rho, z)$ is given by \eqref{eq:g}. 
\end{lemma}
\begin{proof}
From the definition of channel $W^{+}$ in \eqref{align:trans2}, we can write
\begin{align}
E_0( \rho, W_{1,2}^{+}) \nonumber \\
= -\log\displaystyle\sum_{y_1, y_2, u} &\bigg[ \hspace{2mm} \frac{1}{2} W_{1,2}^{+}(y_1, y_2, u \mid 0)^{\frac{1}{1+\rho}} + \frac{1}{2} W_{1,2}^{+}(y_1, y_2, u \mid 1)^{\frac{1}{1+\rho}} \hspace{2mm} \bigg]^{1+\rho} \nonumber \\
= -\log\displaystyle\sum_{y_1, y_2, u} &\bigg[ \hspace{2mm} \frac{1}{2} \left( \frac{1}{2} W_1(y_1 \mid u) W_2(y_2 \mid 0) \right)^{\frac{1}{1+\rho}} + \frac{1}{2} \left( \frac{1}{2} W_1(y_1 \mid u\oplus 1) W_2(y_2 \mid 1)\right)^{\frac{1}{1+\rho}} \bigg]^{1+\rho} \nonumber \\
= -\log\displaystyle\sum_{y_1, y_2} &\left(  \hspace{2mm} \bigg[ \hspace{2mm} \frac{1}{2} \left( \frac{1}{2} W_1(y_1 \mid 0) W_2(y_2 \mid 0)\right)^{\frac{1}{1+\rho}}  +  \frac{1}{2} \left( \frac{1}{2} W_1(y_1 \mid 1) W_2(y_2 \mid 1) \right)^{\frac{1}{1+\rho}}  \hspace{2mm} \bigg]^{1+\rho}\right.  \nonumber \\
&+ \left. \bigg[ \hspace{2mm} \frac{1}{2} \left( \frac{1}{2} W_1(y_1 \mid 1) W_2(y_2 \mid 0)\right)^{\frac{1}{1+\rho}}  +  \frac{1}{2} \left( \frac{1}{2} W_1(y_1 \mid 0) W_2(y_2 \mid 1) \right)^{\frac{1}{1+\rho}}  \hspace{2mm} \bigg]^{1+\rho} \right).  \nonumber 
\end{align}
Using \eqref{eq:dist}, we have
\begin{align}
&E_0(\rho, W_{1,2}^{+}) \nonumber \\
= &-\log \displaystyle\sum_{y_1 y_2} \hspace{2mm} \frac{1}{2} \hspace{2mm} q_{W_1}(y_1) \hspace{2mm} q_{W_2}(y_2) \nonumber \\
&\hspace{16mm} \left( \hspace{2mm} \bigg[ \hspace{2mm} \bigl(\left( 1 + \Delta_{W_1}(y_1)\right) \left( 1 + \Delta_{W_2}(y_2)\right)  \bigr)^{\frac{1}{1+\rho}} +  \bigl(\left( 1 - \Delta_{W_1}(y_1)\right) \left( 1 - \Delta_{W_2}(y_2)\right)  \bigr)^{\frac{1}{1+\rho}} \hspace{2mm} \bigg]^{1+\rho} \right.  \nonumber \\
& \hspace{16mm} + \left. \bigg[ \hspace{2mm} \bigl(\left( 1 - \Delta_{W_1}(y_1)\right) \left( 1 + \Delta_{W_2}(y_2)\right)  \bigr)^{\frac{1}{1+\rho}} +  \bigl(\left( 1 - \Delta_{W_1}(y_1)\right) \left( 1 + \Delta_{W_2}(y_2)\right)  \bigr)^{\frac{1}{1+\rho}} \hspace{2mm} \bigg]^{1+\rho} \right) \nonumber \\
= &-\log \left( \hspace{2mm} \displaystyle\sum_{y_1 y_2} \hspace{2mm} \frac{1}{2} \hspace{2mm} q_{W_1}(y_1) \hspace{2mm} q_{W_2}(y_2) \hspace{2mm} \bigl( 1 + \Delta_{W_1}(y_1) \Delta_{W_2}(y_2)\bigr) \right.\nonumber\\
&\hspace{22mm} \bigg[ \hspace{2mm} \frac{1}{2} \left(1 + \frac{\Delta_{W_1}(y_1) + \Delta_{W_2}(y_2)}{1 + \Delta_{W_1}(y_1)\Delta_{W_2}(y_2)}\right)^{\frac{1}{1+\rho}} + \frac{1}{2} \left(1 - \frac{\Delta_{W_1}(y_1) + \Delta_{W_2}(y_2)}{1 + \Delta_{W_1}(y_1)\Delta_{W_2}(y_2)}\right)^{\frac{1}{1+\rho}}\bigg]^{1+\rho}   \nonumber \\
&\hspace{10mm} + \hspace{2mm} \displaystyle\sum_{y_1 y_2} \hspace{2mm} \frac{1}{2} \hspace{2mm} q_{W_1}(y_1) \hspace{2mm} q_{W_2}(y_2) \hspace{2mm} \bigl( 1 - \Delta_{W_1}(y_1) \Delta_{W_2}(y_2)\bigr) \nonumber \\ 
&\hspace{22mm} \left. \bigg[ \hspace{2mm} \frac{1}{2} \left(1 + \frac{\Delta_{W_1}(y_1) - \Delta_{W_2}(y_2)}{1 - \Delta_{W_1}(y_1)\Delta_{W_2}(y_2)}\right)^{\frac{1}{1+\rho}} + \frac{1}{2} \left(1 - \frac{\Delta_{W_1}(y_1) - \Delta_{W_2}(y_2)}{1 - \Delta_{W_1}(y_1)\Delta_{W_2}(y_2)}\right)^{\frac{1}{1+\rho}}\bigg]^{1+\rho}  \right)   \nonumber \\
&= -\log \left( \hspace{2mm} \displaystyle\sum_{y_1 y_2} \hspace{2mm} \frac{1}{2} \hspace{2mm} q_{W_1}(y_1) \hspace{2mm} q_{W_2}(y_2) \hspace{2mm} \bigl( 1 + \Delta_{W_1}(y_1) \Delta_{W_2}(y_2)\bigr) \hspace{2mm} g\Bigl(\rho, \frac{\Delta_{W_1}(y_1) + \Delta_{W_2}(y_2)}{1 + \Delta_{W_1}(y_1) \Delta_{W_2}(y_2)}\Bigr) \right.\nonumber \\ 
&\hspace{14mm} + \hspace{2mm} \displaystyle\sum_{y_1 y_2} \hspace{2mm} \left. \frac{1}{2} \hspace{2mm} q_{W_1}(y_1) \hspace{2mm} q_{W_2}(y_2) \hspace{2mm} \bigl( 1 - \Delta_{W_1}(y_1) \Delta_{W_2}(y_2)\bigr) \hspace{2mm}  g\Bigl(\rho, \frac{\Delta_{W_1}(y_1) - \Delta_{W_2}(y_2)}{1 - \Delta_{W_1}(y_1) \Delta_{W_2}(y_2)}\Bigr) \right)     \nonumber 
\end{align}
where $g(\rho, z)$ is defined in \eqref{eq:g}.\\\\
Similar to the $E_0(\rho, W_{1,2}^{-})$ case, we define $Z_1 = \lvert\Delta_{W_1}(Y_1)\rvert$ and 
$Z_2 = \lvert\Delta_{W_2}(Y_2)\rvert$ where $Y_1$ and $Y_2$ are independent random variables with distributions $q_{W_1}$ and $q_{W_2}$, respectively.
However, we should check whether this construction is equivalent to the above equation. We note that $\Delta \in [-1, 1]$. 
When $\Delta_{W_1}(y_1)$ and $\Delta_{W_2}(y_2)$ are of the same sign, we can easily see (noting that $g(\rho,z)$ is symmetric about $z=0$) that
\begin{align}
\bigl( 1 + \Delta_{W_1}(y_1) \Delta_{W_2}(y_2)\bigr)\hspace{2mm} g\Bigl( \rho, \frac{\Delta_{W_1}(y_1) + \Delta_{W_2}(y_2)}{1 + \Delta_{W_1}(y_1) \Delta_{W_2}(y_2)}\Bigr)  &=  \bigl( 1 + Z_1 Z_2\bigr)\hspace{2mm}  g\Bigl(\rho, \frac{Z_1 + Z_2}{1 + Z_1 Z_2}\Bigr)  \nonumber \\ 
\bigl( 1 - \Delta_{W_1}(y_1) \Delta_{W_2}(y_2)\bigr)\hspace{2mm} g\Bigl( \rho, \frac{\Delta_{W_1}(y_1) - \Delta_{W_2}(y_2)}{1 - \Delta_{W_1}(y_1) \Delta_{W_2}(y_2)}\Bigr)  &=  \bigl( 1 - Z_1 Z_2\bigr)\hspace{2mm}  g\Bigl(\rho, \frac{Z_1 - Z_2}{1 - Z_1 Z_2}\Bigr) \nonumber 
\end{align}
When $\Delta_{W}(y_1)$ and $\Delta_{W}(y_2)$ are of the opposite sign, we note that
\begin{align}
\bigl( 1 + \Delta_{W_1}(y_1) \Delta_{W_2}(y_2)\bigr)\hspace{2mm}  g\Bigl(\rho, \frac{\Delta_{W_1}(y_1) + \Delta_{W_2}(y_2)}{1 + \Delta_{W_1}(y_1) \Delta_{W_2}(y_2)}\Bigr) =  \bigl( 1 - Z_1 Z_2\bigr)\hspace{2mm}  g\Bigl(\rho, \frac{Z_1 - Z_2}{1 - Z_1 Z_2}\Bigr) \nonumber \\ 
\bigl( 1 - \Delta_{W_1}(y_1) \Delta_{W_2}(y_2)\bigr)\hspace{2mm}  g\Bigl(\rho, \frac{\Delta_{W_1}(y_1) - \Delta_{W_2}(y_2)}{1 - \Delta_{W_1}(y_1) \Delta_{W_2}(y_2)}\Bigr) =  \bigl( 1 + Z_1 Z_2\bigr)\hspace{2mm}  g\Bigl(\rho, \frac{Z_1 + Z_2}{1 + Z_1 Z_2}\Bigr) \nonumber 
\end{align}
Since we are interested in the sum of the above two parts, we can see that the construction we propose is still equivalent. 
This concludes the proof.
\end{proof}

\begin{remark}\label{rem::r1}
 By the symmetry of the RVs $Z_1$ and $Z_2$, we have $E_0( \rho, W_{1,2}^{\pm}) = E_0( \rho, W_{2,1}^{\pm})$.
\end{remark}

\begin{lemma}\label{lem::E0_ordering}
The channels $W_{1,2}^{-}$, $W_1$, $W_2$, and $W_{1,2}^{+}$ satisfy the following ordering:
\begin{align}\label{eq:E0_order}
 &E_0(\rho, W_{1,2}^{-}) \leq E_0(\rho, W_{1}) \leq E_0(\rho, W_{1,2}^{+}), \\
 &E_0(\rho, W_{1,2}^{-}) \leq E_0(\rho, W_{2}) \leq E_0(\rho, W_{1,2}^{+}). \nonumber
\end{align}
\end{lemma}
\begin{proof}
We only show the inequalities in \eqref{eq:E0_order} for the channel $W_1$. The proof for the channel $W_2$ follows from Remark \ref{rem::r1}.
By Lemmas \ref{lem:basic}, \ref{lem:minus}, and \ref{lem:plus} the inequalities in \eqref{eq:E0_order} are equivalent to
\begin{align}
\label{eq::ineq_1}&\Expt\left[ \hspace{2mm} \frac{1}{2} \bigl(1 + Z_1 Z_2\bigr)\hspace{2mm} g\Bigl( \rho, \frac{Z_1 + Z_2}{1 + Z_1 Z_2}\Bigr) +\frac{1}{2} \bigl(1 - Z_1 Z_2\bigr)\hspace{2mm} g\Bigl( \rho, \frac{Z_1 - Z_2}{1 - Z_1 Z_2}\Bigr) \hspace{2mm} \right] \leq \Expt\left[ g(\rho, Z_1)\right], \\
\label{eq::ineq_2}&\Expt\left[ g(\rho, Z_1)\right] \leq \Expt\left[ g(\rho, Z_1 Z_2)\right].
\end{align}
By Lemma \ref{lem:g_concavity}, the function $g(\rho, z)$ is non-increasing in the variable $z$ when $\rho\geq 0$. Hence, the second inequality in \eqref{eq::ineq_2} holds. 
On the other side, note that for any realizations $z_1$ and $z_2$, the factors $\displaystyle\frac{1}{2}(1 + z_1 z_2)$,
and $\displaystyle\frac{1}{2}(1 - z_1 z_2)$ form a distribution. 
As the function $g(\rho, z)$ is concave in $z$ by Lemma \ref{lem:g_concavity}, we can apply Jensen's inequality to obtain
\begin{equation*}
 \frac{1}{2} \bigl(1 + z_1 z_2\bigr)\hspace{2mm} g\Bigl( \rho, \frac{z_1 + z_2}{1 + z_1 z_2}\Bigr) +\frac{1}{2} \bigl(1 - z_1 z_2\bigr)\hspace{2mm} g\Bigl( \rho, \frac{z_1 -z_2}{1 - z_1 z_2}\Bigr) 
\leq g\Bigl( \rho, \frac{z_1 + z_2}{2} + \frac{z_1 - z_2}{2}\Bigr) =  g(\rho, z_1).
\end{equation*}
Taking the expectation of both sides, we get the first inequality in \eqref{eq::ineq_1}.
\end{proof}

\begin{remark}
In \cite{E0_submartingale} it is shown that the channels $W_1$, $W_2$, $W_{1,2}^{-}$, and $W_{1,2}^{+}$ satisfy the relationship:
\begin{equation}
E_0( \rho, W_{1,2}^{+}) + E_0( \rho, W_{1,2}^{-}) \geq E_0(\rho, W_1) + E_0(\rho, W_2), \hspace{3mm} \forall\rho\geq 0. \nonumber 
\end{equation} 
\end{remark}

\subsection{Extremality for the Basic Channel Transformations}
\newtheorem{theorem}{Theorem}
\begin{theorem}\label{thm:extremality}
Given two B-DMCs $W_1$, and $W_2$, for any fixed value of $\rho \geq 0$, we define two binary symmetric channels $W_{\mathsf{BSC}}$, and $W_{\overline{\mathsf{BSC}}}$, and two binary erasure channels $W_{\mathsf{BEC}}$,
and $W_{\overline{\mathsf{BEC}}}$ through the equalities
\begin{align}
\label{eq::Equal_E0_thm2_1} E_0(\rho, W_1) = E_0(\rho, W_{\mathsf{BEC}}) = E_0(\rho, W_{\mathsf{BSC}}), \\
\label{eq::Equal_E0_thm2_2} E_0(\rho, W_2) = E_0(\rho, W_{\overline{\mathsf{BEC}}}) = E_0(\rho, W_{\overline{\mathsf{BSC}}}).
\end{align}
Then for the $W_{1,2}^{-}$ polar transformation, we have
\begin{equation}\label{eq:extremityminus}
E_0(\rho, W_{\mathsf{BEC}, \overline{\mathsf{BEC}}}^{-}) \leq E_0(\rho, W_{1,2}^{-}) \leq E_0(\rho, W_{\mathsf{BSC}, \overline{\mathsf{BSC}}}^{-}) \hspace{5mm} \forall \hspace{1mm}\rho\geq 0.
\end{equation}
For the $W^{+}$ polar transformation, we have
\begin{align}
\label{eq:extremityplus_1} E_0(\rho, W_{\mathsf{BSC},\overline{\mathsf{BSC}}}^{+}) &\leq E_0(\rho, W_{1,2}^{+}) \leq E_0(\rho, W_{\mathsf{BEC}, \overline{\mathsf{BEC}}}^{+}) \hspace{5mm} \forall \hspace{1mm}\rho\in[0, 1]\cup[2, \infty],   \\
\label{eq:extremityplus_2} E_0(\rho, W_{\mathsf{BEC}, \overline{\mathsf{BEC}}}^{+}) &\leq E_0(\rho, W_{1,2}^{+}) \leq E_0(\rho, W_{\mathsf{BSC}, \overline{\mathsf{BSC}}}^{+}) \hspace{5mm} \forall \hspace{1mm}\rho\in[1, 2].
\end{align}
\end{theorem}

\begin{proof}
We start to show the result for the minus transformation given in Equation \eqref{eq:extremityminus}.
This proof relies on the convexity result stated in the next lemma. The proof of the lemma is given in Appendix A. 
\begin{lemma}\label{lem::conv_minus}
For any $z\in[0, 1]$, and $\rho\geq 0$, the function $F_{z, \rho}(t): [2^{-\rho}, 1] \to [g(\rho, z), 1]$ defined as
\begin{equation}\label{eq:convexityminus}
 F_{z, \rho}(t) = g(\rho, z g^{-1}(\rho, t)) 
\end{equation}
where $g^{-1}(\rho, t)$ denotes the inverse of the function $g$ with respect to its second argument, is convex with respect to the variable $t$.
\end{lemma}

From Lemmas \ref{lem:basic}, and \ref{lem:minus}, we know that
\begin{align*}
&\exp\{-E_0(\rho, W_1)\} = \Expt{\left[ g(\rho, Z_1)\right]} \\
&\exp\{-E_0(\rho, W_2)\} = \Expt{\left[ g(\rho, Z_2)\right]} \\
&\exp\{-E_0(\rho, W_{1,2}^{-})\} = \Expt{\left[ g(\rho, Z_1 Z_2)\right]} 
\end{align*}
where $Z_1$ and $Z_2$ are independent random variables. We also know $Z_{\mathsf{BSC}} = z_{\mathsf{BSC}}$, $Z_{\overline{\mathsf{BSC}}} = z_{\overline{\mathsf{BSC}}}$ and 
$Z_{\mathsf{BEC}}, Z_{\overline{\mathsf{BEC}}} \in \left\lbrace 0, 1\right\rbrace$. Hence, 
\begin{equation}
 \exp\{-E_0(\rho, W_{\mathsf{BSC}, \overline{\mathsf{BSC}}}^{-})\} = g(\rho, z_{\mathsf{BSC}} z_{\overline{\mathsf{BSC}}}). \nonumber \\
\end{equation}
Given $E_0(\rho, W_1) = E_0(\rho, W_{\mathsf{BSC}})$, and $E_0(\rho, W_2) = E_0(\rho, W_{\overline{\mathsf{BSC}}})$ we also have
\begin{align*}
&\Expt\left[ g(\rho, Z_1)\right] = g(\rho, z_{\mathsf{BSC}}), \\
&\Expt\left[ g(\rho, Z_2)\right] = g(\rho, z_{\overline{\mathsf{BSC}}}).
\end{align*}
Therefore, using Jensen's inequality we obtain
\begin{align}
\exp\{-E_0(\rho, W_{1,2}^{-})\} &= \Expt_{Z_1}{\left[ \Expt_{Z_2}{\left[ F_{z_1, \rho}\left( g(\rho, Z_2)\right) \mid Z_1 = z_1\right] }\right] } \nonumber \\
&\geq \Expt_{Z_1}{ \left[ F_{Z_1, \rho}\left( \Expt_{Z_2}{ \left[ g(\rho, Z_2)\right] } \right) \right] } \nonumber \\
&= \Expt_{Z_1}{\left[ F_{Z_1, \rho} \left(g(\rho, z_{\overline{\mathsf{BSC}}})\right)\right]} \nonumber \\
&\overset{(1)}{=} \Expt_{Z_1}{\left[ F_{z_{\overline{\mathsf{BSC}}}, \rho} \left(g(\rho, Z_1)\right)\right]} \nonumber \\
&\geq  F_{z_{\overline{\mathsf{BSC}}}, \rho}\left(\Expt_{Z_1}{\left[ g(\rho, Z_1)\right] } \right) \nonumber \\
&= F_{z_{\overline{\mathsf{BSC}}}, \rho}\left(g(\rho, z_{\mathsf{BSC}}) \right) \nonumber \\
&=  \exp\{-E_0(\rho, W_{\mathsf{BSC}, \overline{\mathsf{BSC}}}^{-})\} \nonumber 
\end{align}
where $(1)$ follows by symmetry of the variables $Z_1$ and $z_{\overline{\mathsf{BSC}}}$. 

Let $\epsilon$, and $\overline{\epsilon}$ be the erasure probabilities of $W_{\mathsf{BEC}}$, and $W_{\overline{\mathsf{BEC}}}$, respectively. Then, we have $P(Z_{\mathsf{BEC}} = 0) = \epsilon$, 
$P(Z_{\overline{\mathsf{BEC}}} = 0) = \overline{\epsilon}$, and
\begin{align*}
&\exp\{-E_0(\rho, W_{\mathsf{BEC}})\} = P(Z_{\mathsf{BEC}} = 0)(1 - 2^{-\rho}) + 2^{-\rho},\\
&\exp\{-E_0(\rho, W_{\overline{\mathsf{BEC}}})\} = P(Z_{\overline{\mathsf{BEC}}} = 0)(1 - 2^{-\rho}) + 2^{-\rho}.
\end{align*}
The channel $W_{1,2}^{-}$ is a BEC with erasure probability $\epsilon + \overline{\epsilon}-\epsilon\overline{\epsilon}$, hence we get
\begin{equation}
\exp\{-E_0(\rho, W_{\mathsf{BEC}, \overline{\mathsf{BEC}}}^{-})\} = \left[P(Z_{\mathsf{BEC}} = 0) + P(Z_{\overline{\mathsf{BEC}}} = 0)- P(Z_{\mathsf{BEC}} = 0)P(Z_{\overline{\mathsf{BEC}}} = 0)\right] (1 - 2^{-\rho}) + 2^{-\rho}. \nonumber 
\end{equation}
Therefore, given $E_0(\rho, W_1) = E_0(\rho, W_{\mathsf{BEC}})$, and $E_0(\rho, W_2) = E_0(\rho, W_{\overline{\mathsf{BEC}}})$, we have
\begin{align*}
&\Expt\left[ g(\rho, Z_1)\right] =  \Expt\left[ g(\rho, Z_{\mathsf{BEC}})\right] =  P(Z_{\mathsf{BEC}} = 0) (1 - 2^{-\rho}) + 2^{-\rho}, \\
&\Expt\left[ g(\rho, Z_2)\right] =  \Expt\left[ g(\rho, Z_{\overline{\mathsf{BEC}}})\right] =  P(Z_{\overline{\mathsf{BEC}}} = 0) (1 - 2^{-\rho}) + 2^{-\rho}. 
\end{align*}
Due to convexity, we also know the following inequality holds:
\begin{equation*}
F_{z, \rho}(t) \leq (1 - t)F_{z, \rho}(0) + tF_{z, \rho}(1) = 1 + \frac{g(\rho, z) - 1}{2^{-\rho} - 1} (t - 1). 
\end{equation*}
Therefore,
\begin{align}
&\hspace{3mm}\exp\{-E_0(\rho, W_{1, 2}^{-})\} \\
&= \Expt_{Z_1}{\left[ \Expt_{Z_2}{\left[ F_{z_1, \rho}\left( g(\rho, Z_2)\right) \mid Z_1 = z_1\right] }\right] } \nonumber \\
&\leq \Expt_{Z_1}{\left[ 1 + \frac{g(\rho, Z_1) - 1}{2^{-\rho} - 1} (\Expt_{Z_2}{\left[ g(\rho, Z_2)\right]}  - 1) \right]} \nonumber  \\
&= 1 + \frac{\Expt_{Z_1}{\left[ g(\rho, Z_1) \right]  - 1}}{2^{-\rho} - 1} (\Expt_{Z_2}{\left[ g(\rho, Z_2)\right]}  - 1) \nonumber \\
&= 1 + \frac{\left[ P(Z_{\mathsf{BEC}} = 0) (1 - 2^{-\rho}) + 2^{-\rho} - 1\right]\left[ P(Z_{\overline{\mathsf{BEC}}} = 0) (1 - 2^{-\rho}) + 2^{-\rho} - 1 \right]}{2^{-\rho} - 1} \nonumber \\
&= 1 - P(Z_{\mathsf{BEC}} = 0) P(Z_{\overline{\mathsf{BEC}}} = 0)(1 - 2^{-\rho}) +  \left( P(Z_{\mathsf{BEC}} = 0) + P(Z_{\overline{\mathsf{BEC}}} = 0)\right)(1 - 2^{-\rho}) + 2^{-\rho}-1 \nonumber \\
&= \left[ P(Z_{\mathsf{BEC}} = 0) + P(Z_{\overline{\mathsf{BEC}}} = 0) - P(Z_{\mathsf{BEC}} = 0)P(Z_{\overline{\mathsf{BEC}}} = 0)\right] (1 - 2^{-\rho}) + 2^{-\rho} \nonumber \\
&= \exp\{-E_0(\rho, W_{\mathsf{BEC}, \overline{\mathsf{BEC}}}^{-})\}.  \nonumber 
\end{align}
This concludes the proof for the minus transformation. Now, we sketch the proof of the extremality property for the plus transformation.
We define the function $h(\rho, z_1, z_2)$ as 
\begin{equation}\label{eq:h}
 h(\rho, z_1, z_2) = \frac{1}{2} \bigl(1 + z_1 z_2\bigr) g\Bigl(\rho, \frac{z_1 + z_2}{1 + z_1 z_2}\Bigr) + \frac{1}{2} \bigl(1 - z_1 z_2\bigr) g\Bigl(\rho, \frac{z_1 - z_2}{1 - z_1 z_2}\Bigr) 
\end{equation}
where $z_1, z_2\in[0, 1]$, and $\rho\geq 0$. Note that $h(\rho, z_1, z_2$) is symmetric in the variables $z_1$, and $z_2$.
The proof relies on the convexity result stated in the next lemma. The proof of the lemma is given in Appendix B. 
\begin{lemma}\cite{private:Emre_Telatar}\label{lem::conv_plus}
For any $z\in[0, 1]$, and $\rho\geq 0$, the function $H_{z, \rho}(t): [2^{-\rho}, 1] \to [2^{-\rho}, g(\rho, z)]$ defined as
\begin{equation*}\label{eq:convexityplus}
H_{z, \rho}(t) = h(\rho, g^{-1}(\rho, t), z) 
\end{equation*}
is concave with respect to the variable $t$ when $\rho\in[0, 1]\cup[2, \infty]$, and convex when $\rho\in[1, 2]$.
\end{lemma}

The proof of the theorem for the plus transformation can be completed following similar steps to the minus case.
By Lemma \ref{lem:plus}, we have
\begin{equation*}
\Expt\left[h(\rho, Z_1, Z_2)\right] = \exp\{-E_0(\rho, W_{1,2}^{+})\}. 
\end{equation*}
We define the random variables
\begin{equation*}
T_1 = g(\rho, Z_1) \hspace{3mm} \hbox{and} \hspace{3mm} T_2 = g(\rho, Z_2).
\end{equation*}
Then, using the concavity of the function $H_{z, \rho}(t)$ with respect to $t$ for fixed values of $\rho\in [0, 1] \cup [2, \infty]$, and $z\in[0, 1]$, we obtain the inequalities in \eqref{eq:extremityplus_1}: 
\begin{equation*}
 \exp\{-E_0(\rho, W_{1,2}^{+})\} = \Expt\left[H_{g^{-1}(\rho, T_2), \rho}(T_1) \right] \leq h(\rho, z_{\mathsf{BSC}}, z_{\overline{\mathsf{BSC}}}) = \exp\{-E_0(\rho, W_{\mathsf{BSC}, \overline{\mathsf{BSC}}}^{+})\},
\end{equation*}
and 
\begin{multline*}
 \exp\{-E_0(\rho, W_{1.2}^{+})\} = \Expt\left[H_{g^{-1}(\rho, T_2), \rho}(T_1) \right] \geq  2^{-\rho} + P(Z_{\mathsf{BEC}} = 0)P(Z_{\overline{\mathsf{BEC}}} = 0) \left(1 - 2^{-\rho} \right) \\
= \exp\{-E_0(\rho, W_{\mathsf{BEC}, \overline{\mathsf{BEC}}}^{+})\}.
\end{multline*}
Similarly, the convexity of the function $H_{z, \rho}(t)$ with respect to $t$ for $\rho\in[1, 2]$ leads to the reverse inequalities in \eqref{eq:extremityplus_2}.
\end{proof}

\subsection{Special $\rho$ Values}
In Theorem \ref{thm:extremality}, we have shown that among all B-DMC's $W$ of fixed $E_0(\rho, W)$, the binary erasure channel's minus transformation 
results in a lower bound to any $E_0(\rho, W^{-})$ and the binary symmetric channel's one in an upper bound to any $E_0(\rho, W^{-})$. 
For the plus transformation, a similar extremality property holds except the difference that the result breaks into two parts depending on the value of the parameter $\rho$:
While the binary erasure and binary symmetric channels appear on opposite sides of the inequalities
for $E_0(\rho, W^{-})$ and $E_0(\rho, W^{+})$ when $\rho\in[0, 1]\cup[2, \infty]$, they appear on the same side when $\rho\in[1, 2]$. 
Using these results, we identify in this section some special cases of $\rho$ values to recover known, and discover new results. 

\subsubsection{$\rho = 0$, Symmetric capacity}
In \cite{1669570}, it is shown that the symmetric capacity is preserved under the basic polarization transformations. 
This property holds regardless of whether the combined channels are identical or not, as it is a consequence of the chain rule for mutual information. Namely, the channels satisfy:
\begin{equation*}
 2I(W) = I(W^{-}) + I(W^{+}). 
\end{equation*}
This relation implies the process attached to the symmetric capacities of the synthesized channels is a bounded martingale, hence converges almost surely. 
\newtheorem{corollary}{Corollary}
\begin{corollary}\label{cor::E0_polariyation_pm_diff_extremality}
Under the assumptions as Theorem~\ref{thm:extremality} with $W_1 = W_2 = W$, we have
\begin{equation*}
E_0(\rho, W_{\mathsf{BSC}}^{+}) - E_0(\rho, W_{\mathsf{BSC}}^{-})
\leq E_0(\rho, W^{+}) - E_0(\rho, W^{-}) 
\leq E_0(\rho, W_{\mathsf{BEC}}^{+}) - E_0(\rho, W_{\mathsf{BEC}}^{-}).\\
\end{equation*}
for $\rho\in[0, 1]$
\end{corollary}
Corollary \ref{cor::E0_polariyation_pm_diff_extremality} shows that amongst channels $W$ with a given value of
$E_0(\rho, W)$ for a given $\rho$ the BEC and BSC are the most and least
polarizing under Ar{\i}kan's polar transformations in the sense that their
polar transforms $W^+$ and $W^-$ has the largest and smallest difference
in their $E_0$ values. Dividing all sides of the inequality above by $\rho$ and taking the limit
as $\rho\to0$, we see that among channels of a given symmetric capacity,
the BEC and BSC are extremal with respect to the polarization transformations,
in the sense that
\begin{equation*}
     I(W_{\mathsf{BSC}}^{+}) - I(W_{\mathsf{BSC}}^{-})
\leq I(W^{+}) - I(W^{-})
\leq I(W_{\mathsf{BEC}}^{+}) - I(W_{\mathsf{BEC}}^{-}).
\end{equation*}
This is a known argument proving the convergence is to the extremes of the $[0, 1]$ interval. 
The preservation property of the symmetric capacities holds regardless of whether the combined channels are identical or not, as it is a consequence of the chain rule for mutual information. 
Namely, the channels satisfy:
\begin{equation*}
 I(W_1) + I(W_2) = I(W_{1, 2}^{-}) + I(W_{1, 2}^{+}),
\end{equation*}
and Theorem \ref{thm:extremality} can be used to show the convergence is also to the extremes values $\{0, 1\}$ of the corresponding bounded martingale process.

\begin{remark}
These inequalities for the symmetric capacities can also be obtained by the results on the extremes of information combining \cite{1412027}, 
together with the fact that symmetric capacity is preserved under the polarization transformations \cite{1669570}.
\end{remark}

\subsubsection{$\rho = 1,$ Cut-off rate, Bhatthacharyya parameter}
Another result of \cite{1669570} can be recovered by letting $\rho=1$. In this case, Theorem \ref{thm:extremality} implies channels having equal cut-off rates satisfy
\begin{align*}
 E_0(1, W_{\mathsf{BEC}}^{-}) \leq E_0(1, W^{-}) \leq E_0(1, W_{\mathsf{BSC}}^{-}), \\
 E_0(1, W_{\mathsf{BSC}}^{+}) = E_0(1, W^{+}) =  E_0(1, W_{\mathsf{BEC}}^{+}).
\end{align*}
Moreover, by the definition in Equation \eqref{eq::Z_def}, the extremalities for the Bhattacharyya parameter are also obtained. Indeed, we know $Z(W^{+}) = Z(W)^2$ by \cite{1669570}.
\subsubsection{$\rho = 2$}
A previously unknown result is found by taking $\rho = 2$ in the theorem. Similar to the case $\rho = 1$, we observe the $E_0$ parameter of the channels $W^{+}, W_{\mathsf{BEC}}^{+}, W_{\mathsf{BSC}}^{+}$ are 
equal to each other. 

\subsection{Generalizations of the Bhatthacharyya parameter}
In this section, we discuss a generalization to the definition of the Bhattacharyya parameter. 
We propose an extension motivated by the $E_0$ parameter of BECs. Given a BEC $W_{\mathsf{BEC}}$ with erasure probability $\epsilon_{\mathsf{bec}}$, we have
\begin{equation}
\epsilon_{\mathsf{bec}} = \frac{2^{\rho} 2^{-E_0(\rho, W_{\mathsf{BEC}})} - 1}{2^{\rho} - 1}. \nonumber 
\end{equation}
We also know the Bhattacharyya parameter of a binary erasure channel satisfies $Z(W_{\mathsf{BEC}}) = \epsilon_{\mathsf{bec}}$. 
This parameter provides tighter bounds than $E_0(1, W)$ in~\cite{1669570}, and is used in the subsequent analysis. This gives the idea
to define a similar quantity to $Z(W)$, referred as $Z(\rho, W)$, which reflects the dependence on the value of $\rho$  
\begin{equation}
Z(\rho, W) = \frac{2^{\rho} 2^{-E_0(\rho, W)} - 1}{2^{\rho} - 1}. \nonumber
\end{equation}
Using the results we derived in the previous section, the next Corollary shows how $Z(\rho, W)$ is affected by the basic channel transformations.
\begin{corollary}\label{cor::Zs}
Given a B-DMC $W$, for any fixed value of $\rho \geq 0$, we define a binary symmetric channel $W_{\mathsf{BSC}}$, and a binary erasure channel $W_{\mathsf{BEC}}$ through
the equality 
\begin{equation}
 Z(\rho, W) = Z(\rho, W_{\mathsf{BEC}}) = Z(\rho, W_{\mathsf{BSC}}) \nonumber
\end{equation}
Then for the $W^{-}$ and $W^{+}$ polar transformations, we have
\begin{align}
Z(\rho, W_{\mathsf{BSC}}^{-}) \hspace{2mm} \leq \hspace{2mm} &Z(\rho, W^{-}) \hspace{2mm} \leq \hspace{2mm} Z(\rho, W_{\mathsf{BEC}}^{-}) = 2 Z(\rho, W_{\mathsf{BEC}}) - Z(\rho, W_{BEC})^2, \nonumber \\
Z(\rho, W_{\mathsf{BEC}})^2 = Z(\rho, W_{\mathsf{BEC}}^{+}) \hspace{2mm} \leq \hspace{2mm} &Z(\rho, W^{+}) \hspace{2mm} \leq \hspace{2mm} Z(\rho, W_{\mathsf{BSC}}^{+})  \hspace{8mm} \forall \hspace{1mm}\rho\in[0, 1]\cup[2, \infty], \nonumber \\
Z(\rho, W_{\mathsf{BSC}}^{+}) \hspace{2mm} \leq \hspace{2mm} &Z(\rho, W^{+}) \hspace{2mm} \leq \hspace{2mm} Z(\rho, W_{\mathsf{BEC}}^{+}) = Z(\rho, W_{\mathsf{BEC}})^2  \hspace{2mm} \forall \hspace{1mm}\rho\in[1, 2].
\end{align}
\end{corollary}

\section{Conclusions}
The extremality of the BEC and BSC for polar transforms can be interpreted in the context of information combining. 
Theorem \ref{thm:extremality} shows that even if we change the measure of information from the customary mutual
information to $E_0$ the channels BEC and BSC still remain extremal. 
The results of the theorem also show the $\rho = 1, 2$ values share a common property: 
One can recover the value of the parameter $E_0(\rho, W^{+})$ from the value of $E_0(\rho, W)$ without necessarily knowing the particular channel $W$. 
Finally, the extremality results of the theorem open up the possibility to apply the theory of channel polarization to combining arbitrary B-DMCs, 
the details of which will further be investigated in a future work.

\section{Acknowledgment}
The author would like to thank Emre Telatar for helpful discussions. 
This work was supported by Swiss National Science Foundation under grant number 200021-125347/1.

\section*{Appendices}
In these appendices, we prove in part A Lemma \ref{lem::conv_minus}, and in part B Lemma \ref{lem::conv_plus}. For the proofs, we need the following lemma.

\begin{lemma}\label{lem:g_concavity}
The function $g(\rho, z)$ defined as 
\begin{equation*}
 g(\rho, z) =  \left(  \frac{1}{2}\left(1 + z \right)^{\frac{1}{1+\rho}} + \frac{1}{2}\left(1 - z\right)^{\frac{1}{1+\rho}}  \right) ^{1+\rho},
\end{equation*}
for $z\in[0, 1]$, and $\rho\in\mathbf{R}\setminus\{-1\}$, is a concave non-increasing function in $z$ for $\rho\in(-\infty, -1)\cup[0, \infty)$, 
and a convex non-decreasing function in $z$ for $\rho\in(-1, 0]$.
\end{lemma}

\begin{proof}
Taking the first derivative with respect to $z$, we get
\begin{align}
\frac{\partial g(\rho, z)}{\partial z} 
&=  \left( \frac{1}{2}(1 + z)^{\frac{1}{1+\rho}} + \frac{1}{2}(1 - z)^{\frac{1}{1+\rho}} \right)^{\rho}  \left( \frac{1}{2}(1 + z)^{\frac{-\rho}{1+\rho}} - \frac{1}{2}(1 - z)^{\frac{-\rho}{1+\rho}} \right) \nonumber \\
&= \underbrace{\left(\frac{1}{2}\right)^{1+\rho} \left(1 + \left( \frac{1 - z}{1 + z}\right) ^{\frac{1}{1+\rho}}\right)^{\rho}}_{\geq 0} \left( 1 - \left( \frac{1 - z}{1 + z}\right) ^{\frac{-\rho}{1+\rho}}\right) \label{eq::dgu}.
\end{align}
As we have 
\begin{equation*}
\displaystyle\frac{1 - z}{1 + z} \leq 1,  
\end{equation*}
for $\forall z\in[0, 1]$, the monotonicity claims follow by noting that when $\rho\in(-\infty, -1)\cup[0, \infty)$:
\begin{equation*}
\frac{\rho}{1+\rho} \geq 0 \quad \Rightarrow \quad \left( 1 - \left( \frac{1 - z}{1 + z}\right) ^{\frac{-\rho}{1+\rho}}\right) \leq 0 \quad \Rightarrow \quad \frac{\partial g(\rho, z)}{\partial z} \leq 0,
\end{equation*}
and when $\rho\in(-1, 0]$:
\begin{equation*}
 \frac{\rho}{1+\rho} \leq 0 \quad \Rightarrow \quad \left( 1 - \left( \frac{1 - z}{1 + z}\right) ^{\frac{-\rho}{1+\rho}}\right) \geq 0 \quad \Rightarrow \quad \frac{\partial g(\rho, z)}{\partial z} \geq 0. 
\end{equation*}

Taking the second derivative with respect to $z$, we get
\begin{equation}
 \frac{\partial^{2} g(\rho, z)}{\partial z^{2}} = - \frac{\rho}{1+\rho} \underbrace{\left( 1 - z^{2}\right)^{\frac{1}{1+\rho}-2} \left(  \frac{1}{2}(1 + z)^{\frac{1}{1+\rho}} + \frac{1}{2}(1 - z)^{\frac{1}{1+\rho}}  \right) ^{-1+\rho}}_{\geq 0}. \nonumber
\end{equation}
The convexity claims follow once more by inspecting the sign of $\displaystyle\frac{\rho}{1+\rho}$ in different intervals, i.e. when $\rho\in(-\infty, -1)\cup[0, \infty)$:
\begin{equation*}
\frac{\rho}{1+\rho} \geq 0 \quad \Rightarrow \quad \frac{\partial^{2} g(\rho, z)}{\partial z^{2}} \leq 0, 
\end{equation*}
and when $\rho\in(-1, 0]$:
\begin{equation*}
 \frac{\rho}{1+\rho} \leq 0 \quad \Rightarrow \quad \frac{\partial^{2} g(\rho, z)}{\partial z^{2}} \geq 0. 
\end{equation*}
\end{proof}

\subsection*{Appendix A}\label{appendix:A}
\begin{proof}[Proof of Lemma \ref{lem::conv_minus}]
We prove that the function $ F_{z, \rho}(t) = g(\rho, z g^{-1}(\rho, t)) $ defined in Equation \ref{eq:convexityminus}
is convex with respect to the variable $t$ for fixed $\rho\geq 0$ and $z\in[0, 1]$ values.
Taking the first derivative with respect to $t$, we obtain 
\begin{align}
\frac{\partial F_{z, \rho}(t)}{\partial t} &= \frac{\partial}{\partial t}g(\rho, zg^{-1}(\rho, t)) \nonumber \\
&= \frac{g'(\rho, z g^{-1}(\rho_1, t))}{g'(\rho_1, g^{-1}(\rho_1, t))} z. \nonumber
\end{align}
We define $u = g^{-1}(\rho, t)$. Since $g(\rho, u)$ is a non-increasing function in $u$ when $\rho\geq 0$ by Lemma \ref{lem:g_concavity}, 
so is $g^{-1}(\rho, t)$ in $t$. Hence we can check the convexity of $F_{z, \rho}(t)$ with respect to the variable $t$,
from the monotonicity with respect to $u$ of the following expression:
\begin{equation}\label{eq:mono}
 z\frac{g'(\rho_2, zu)}{g'(\rho_1, u)}.
\end{equation}

To simplify notation, we define
\begin{align}
 \label{eq:f} f(u) &= \frac{1 - u}{1 + u}  \\
 \label{eq:alpha}\alpha(\rho, u) &= (1 +  f(u)^{\frac{1}{1+\rho}})^{\rho}  \geq 0\\
 \label{eq:beta}\beta(\rho, u) &= (1 - f(u)^{\frac{-\rho}{1+\rho}}) \leq 0 
\end{align}
Then, by equation \eqref{eq::dgu}
\begin{equation}
 \frac{\partial g(\rho, u)}{\partial u} =  \left(\frac{1}{2}\right)^{1+\rho} \alpha(\rho, u) \beta(\rho, u) \nonumber
\end{equation}
Similarly,
\begin{align}
 \frac{\partial g(\rho, zu)}{\partial u} &= z g'(\rho, zu) \nonumber \\ 
&=  \left(\frac{1}{2}\right)^{1+\rho} z \alpha(\rho, zu) \beta(\rho, zu),  \nonumber
\end{align}
and \eqref{eq:mono} is given by
\begin{equation}\label{eq:appBexp1}
 z \frac{g'(\rho_2, zu)}{g'(\rho_1, u)} = z \frac{\alpha(\rho, zu) \beta(\rho, zu)}{\alpha(\rho, u) \beta(\rho, u)}. 
\end{equation}
Now taking the derivative of \eqref{eq:appBexp1} with respect to $u$, we get
\begin{align}
&\frac{\partial}{\partial u} \hspace{3mm} z \frac{\alpha(\rho,zu) \beta(\rho, zu)}{\alpha(\rho, u) \beta(\rho, u)} \nonumber \\
= &z \underbrace{\frac{\alpha(\rho, zu) \beta(\rho, zu)}{\alpha(\rho, u) \beta(\rho, u)}}_{\geq 0} \nonumber \\
\label{eq:sign1}&\left( \frac{\partial \alpha(\rho, zu)/\partial u}{\alpha(\rho, zu)} + \frac{\partial \beta(\rho, zu)/\partial u}{\beta(\rho, zu)} - \frac{\partial \alpha(\rho, u)/\partial u}{\alpha(\rho, u)} - \frac{\partial \beta(\rho, u)/\partial u}{\beta(\rho, u)}\right) 
\end{align}
We can see that the sign of the expression inside the parenthesis in \eqref{eq:sign1} will determine the monotonicity in $u$ of the expression in 
\eqref{eq:appBexp1}. At this point, we note that
\begin{equation}\label{eq:monotonicity1}
 \frac{\partial \alpha(\rho, u)/\partial u}{\alpha(\rho, u)} + \frac{\partial \beta(\rho, u)/\partial u}{\beta(\rho, u)} =  \left( \frac{\partial \alpha(\rho, zu)/\partial u}{\alpha(\rho, zu)} + \frac{\partial \beta(\rho, zu)/\partial u}{\beta(\rho, zu)}\right)\Bigr\rvert_{z = 1} 
\end{equation}
Moreover, we claim that the expression inside the parenthesis in the RHS of \eqref{eq:monotonicity1} is non-decreasing in $z$. As a consequence, $F_{z, \rho}(t)$ is a concave function 
in $u = g^{-1}(\rho, t)$. Since $u$ is decreasing in $t$, we have  
\begin{equation}
 \displaystyle\frac{\partial^2 F_{z, \rho}(t)}{\partial t^2} = \underbrace{\frac{\partial}{\partial u}\left( z \frac{g'(\rho_2, zu)}{g'(\rho_1, u)} \right)}_{\leq 0} \underbrace{\frac{\partial u}{\partial t}}_{\leq 0} \geq 0. \nonumber 
\end{equation}
We conclude that $F_{z, \rho}(t)$ is a convex function with respect to variable $t$. 

In the rest of the appendix, we prove our claim. We have,
\begin{align} 
\label{eq:der_alphah_u}\frac{\partial \alpha(\rho, zu)}{\partial u} &= \frac{\rho}{1+\rho}  z f'(zu) f(zu)^{\frac{-\rho}{1+\rho}} (1 +  f(zu)^{\frac{1}{1+\rho}})^{\rho-1} \\
\label{eq:der_betah_u} \frac{\partial \beta(\rho, zu)}{\partial u} &= \frac{\rho}{1+\rho} z f'(zu) f(zu)^{\frac{-\rho}{1+\rho} - 1} 
\end{align}
where
\begin{equation}
 f'(u) = \displaystyle\frac{\partial f(u)}{\partial u} = \frac{-2}{(1 + u)^2}. \nonumber 
\end{equation}
Hence,
\begin{align}
 &\frac{\partial \alpha(\rho, zu)/\partial u}{\alpha(\rho, zu)} + \frac{\partial \beta(\rho, zu)/\partial u}{\beta(\rho, zu)} \nonumber \\
=\hspace{2mm} &\frac{\rho}{1+\rho} f(zu)^{\frac{-\rho}{1+\rho}-1} z f'(zu) \left( \frac{f(zu)}{1 + f(zu)^{\frac{1}{1+\rho}}} + \frac{1}{1 -f(zu)^{\frac{-\rho}{1+\rho}}}\right)   \nonumber \\
=\hspace{2mm}  &\frac{\rho}{1+\rho} f(zu)^{\frac{-\rho}{1+\rho}-1} z f'(zu) \left( \frac{f(zu) - f(zu)^{\frac{1}{1+\rho}} + 1  + f(zu)^{\frac{1}{1+\rho}} }{(1 + f(zu)^{\frac{1}{1+\rho}}) (1 -f(zu)^{\frac{-\rho}{1+\rho}})} \right)  \nonumber \\
=\hspace{2mm}  &\frac{\rho}{1+\rho} f(zu)^{\frac{-\rho}{1+\rho}-1} z f'(zu) (1 + f(zu)) (1 + f(zu)^{\frac{1}{1+\rho}})^{-1} (1 -f(zu)^{\frac{-\rho}{1+\rho}})^{-1} \nonumber \\
=\hspace{2mm}  &\frac{\rho}{1+\rho} z f'(zu) (1 + f(zu)^{-1}) (1 + f(zu)^{\frac{1}{1+\rho}})^{-1} (f(zu)^{\frac{\rho}{1+\rho}} - 1)^{-1} \nonumber \\
=\hspace{2mm}  &\frac{\rho}{1+\rho} \frac{-4 z}{(1 + zu)^2 (1 - zu)} \left( 1 + \left( \frac{1 - zu}{1 + zu}\right) ^{\frac{1}{1+\rho}}\right)^{-1}  \left( -1 + \left( \frac{1 - zu}{1 + zu}\right)^{\frac{\rho}{1+\rho}}\right)^{-1} \nonumber \\
=\hspace{2mm}  &\frac{4\rho}{1+\rho}  \left(  \underbrace{\frac{1 - z^2u^2}{z} \left( (1 + zu)^{\frac{\rho}{1+\rho}} - (1 - zu)^{\frac{\rho}{1+\rho}}\right)}_{\hbox{Part 2}}  \underbrace{\left( (1 + zu)^{\frac{1}{1+\rho}} + (1 - zu)^{\frac{1}{1+\rho}}\right)}_{\hbox{Part 1}} \right) ^{-1} \nonumber
\end{align}
We consider the expressions labeled as Part 1 and Part 2 separately. Note that both are positive valued. In addition, we will show that 
both are decreasing in $z$. As a result, we deduce
\begin{align}
&\frac{\partial}{\partial z}   \left(\frac{ \left( 1 - z^2u^2 \right)}{z}  \left( \left( 1 + zu \right) ^{\frac{1}{1+\rho}} +  \left( 1 - zu \right) ^{\frac{1}{1+\rho}} \right) \left( \left( 1 + uz \right) ^{\frac{\rho}{1+\rho}} -  \left( 1 - uz \right) ^{\frac{\rho}{1+\rho}}\right)\right)   \leq 0 \nonumber \\
&\frac{\partial}{\partial z} \left( \frac{\left(1 - z^2u^2\right)}{z} \left(\left(1 + zu\right)^{\frac{1}{1+\rho}} + \left(1 - zu\right)^{\frac{1}{1+\rho}}\right) \left(\left(1 + uz\right)^{\frac{\rho}{1+\rho}} - \left(1 - uz\right)^{\frac{\rho}{1+\rho}}\right)\right)^{-1}  \geq 0 \nonumber 
\end{align}
which is proves our claim. 

For Part 1, we get
\begin{align}
&\frac{\partial}{\partial z} \left(\left(1 + zu\right)^{\frac{1}{1+\rho}} + \left(1 - zu\right)^{\frac{1}{1+\rho}}\right) \nonumber \\
= &\frac{u \left(\left(1 + uz\right)^{\frac{-\rho}{1+\rho}} - \left(1 - uz\right)^{\frac{-\rho}{1+\rho}}\right)}{1+\rho} \leq 0 \nonumber 
\end{align}

For Part 2, we have
\begin{align}
&\frac{\partial}{\partial z} \left(\frac{\left(1 - u^2z^2\right)}{z} \left(\left(1 + uz\right)^{\frac{\rho}{1+\rho}} - \left(1 - uz\right)^{\frac{\rho}{1+\rho}}\right)\right) \nonumber \\
= &\hspace{2mm} \frac{1}{z^2} \frac{\rho u z \left( 1 - u^2 z^2\right) \left( \left( 1 + uz\right)^{\frac{\rho}{1+\rho}-1} + \left( 1 - uz\right)^{\frac{\rho}{1+\rho}-1} \right)}{1+\rho} \nonumber \\
\hspace{1mm}+ &\hspace{1mm}\frac{1}{z^2} \left( 1 + u^2 z^2\right) \left(-\left( 1 + uz\right)^{\frac{\rho}{1+\rho}} + \left( 1 - uz\right)^{\frac{\rho}{1+\rho}}\right)  \nonumber \\
= &\frac{1}{z^2} \left(\hspace{2mm} \left( 1 + uz\right)^{\frac{\rho}{1+\rho}} \left(\frac{\rho}{1+\rho} uz \left( 1 - uz\right) - (1 + u^2 z^2) \right)\right.  \nonumber \\
&\hspace{5mm} + \left. \left( 1 - uz\right)^{\frac{\rho}{1+\rho}} \left(\frac{\rho}{1+\rho} uz \left( 1 + uz\right) + (1 + u^2 z^2) \right)   \right) \nonumber \\ 
=  &\frac{1}{z^2} \left(\hspace{2mm}-\left( 1+x\right)^k \left( \left( k+1\right)x^2 - k x + 1\right) + \left( 1-x\right)^k \left( \left( k+1\right)x^2 + k x + 1\right) \right) \nonumber \\
\label{eq:appBexp2} = &\frac{1}{z^2} \left(\hspace{2mm}-f_1(x, k) + f_2(x, k) \right)
\end{align}
where $k = \frac{\rho}{1+\rho} \in[0, 1)$, $x = uz \in[0, 1]$, and
\begin{align}
 \label{eq:f1}f_1(x, k) &= \left( 1+x\right)^k \left( \left( k+1\right)x^2 - k x + 1\right), \\
 \label{eq:f2}f_2(x, k) &= \left( 1-x\right)^k \left( \left( k+1\right)x^2 + k x + 1\right). 
\end{align}

We will show that $f_1(x, k) \geq f_2(x, k)$ holds for $x\in[0, 1]$, and for $k\in[0, 1)$. Since $f_1(x, k), f_2(x, k) \geq 0$, 
this is equivalent to showing that $\log\displaystyle\frac{f_1(x, k)}{f_2(x, k)} \geq 0$ holds. We have
\begin{equation*}
 \log\displaystyle\frac{f_1(x, k)}{f_2(x, k)} = k  \log\displaystyle\frac{1+x}{1-x} + \log\left( \left( k+1\right)x^2 - k x + 1\right) - \log\left( \left( k+1\right)x^2 + k x + 1\right).
\end{equation*}
We immediately observe that when $k = 0$ we have the above sum equals to 0. Now, we will show that 
\begin{equation*}
 \displaystyle\frac{\partial}{\partial k}\log\displaystyle\frac{f_1(x, k)}{f_2(x, k)} \geq 0.
\end{equation*}
Hence, this will prove our claim that $f_1(x, k) \geq f_2(x, k)$ holds.

Taking the first derivative with respect to $k$, we have
\begin{equation*}
 \displaystyle\frac{\partial}{\partial k}\log\displaystyle\frac{f_1(x, k)}{f_2(x, k)} = \log\displaystyle\frac{1+x}{1-x} - \displaystyle\frac{2x\left(1 + x^2\right)}{\left(1 + (k+1)x^2\right)^2 - \left(kx\right)^2}.
\end{equation*}
So, we will be done if
\begin{equation*}
 \log\displaystyle\frac{1+x}{1-x} \geq 2x\left(1 + x^2\right) \displaystyle\max_{k\in[0, 1)} \displaystyle\frac{1}{\left(1 + (k+1)x^2\right)^2 - \left(kx\right)^2}.
\end{equation*}
One can easily check that the expression in the denominator $\left(1 + (k+1)x^2\right)^2 - \left(kx\right)^2$ is non-decreasing in $k\in[0, 1)$, hence the reciprocal is non-increasing in $k$. As a result, the 
maximum is attained at $k = 0$. Therefore, we only have to prove that
\begin{equation*}
 \log\displaystyle\frac{1+x}{1-x} \geq \displaystyle\frac{2x\left(1 + x^2\right)}{\left(1 + x^2\right)^2} = \displaystyle\frac{2x}{\left(1 + x^2\right)}
\end{equation*}
holds. But, we have
\begin{equation*}
 \log\displaystyle\frac{1+x}{1-x} = 2x\left(1 + \displaystyle\frac{1}{3}x^2 + \displaystyle\frac{1}{5}x^4 + \displaystyle\frac{1}{7}x^6 + \ldots\right) \geq 2x \geq \displaystyle\frac{2x}{\left(1 + x^2\right)}.
\end{equation*}
So, $-f_1(x, k)+ f_2(x, k) \leq 0$ holds for $k\in[0, 1)$ and $x\in[0, 1]$. Consequently, Part 2 is also decreasing in $z$. This proves our claim that the RHS of \eqref{eq:monotonicity1} is non-decreasing in $z$.
\end{proof}

\subsection*{Appendix B}\label{appendix:B}
\begin{proof}[Proof of Lemma \ref{lem::conv_plus}]
In this Appendix, we show that the function $H_{z, \rho}(t) = h(\rho, g^{-1}(\rho, t), z)$ defined in Equation \ref{eq:convexityplus} 
is concave with respect to the variable $t$ when $\rho\in[0, 1]\cup[2, \infty]$, and convex otherwise when $\rho\in[1, 2]$, for any fixed $z\in[0, 1]$, and $\rho\geq 0$.

Taking the first derivative with respect to $t$, we get
\begin{equation*}
 \displaystyle\frac{\partial}{\partial\ t}H_{z, \rho}(t) = \displaystyle\frac{ h'(\rho, g^{-1}(\rho, t), z)}{g'(\rho, g^{-1}(\rho, k))}.
\end{equation*}
As we did in Appendix A, we define $u = g^{-1}(\rho_1, t)$. Since $g(\rho, u)$ is a non-increasing function in $u$ by Lemma \ref{lem:g_concavity}, so is $g^{-1}(\rho_1, t)$ in $t$. 
Hence we can check the concavity of $H_{z, \rho_1}(t)$ with respect to variable $t$, by verifying that 
\begin{equation*}
 \displaystyle\frac{h'(\rho, u, z)}{g'(\rho, u)}
\end{equation*}
is non-decreasing in $u$.
So, we check that 
\begin{equation*}
 \displaystyle\frac{\partial}{\partial u} \left(\displaystyle\frac{h'(\rho, u, z)}{g'(\rho, u)}\right) = \displaystyle\frac{h''(\rho, u, z)g'(\rho, u) - h'(\rho, u, z)g''(\rho, u)}{g'(\rho, u)^2} \geq 0.
\end{equation*}
Since the denominator is always positive, we only need to show that
\begin{equation}
h''(\rho, u, z)g'(\rho, u) - h'(\rho, u, z)g''(\rho, u) \geq 0.
\end{equation}
Moreover, we observe that $h(\rho, u, 0) = g(\rho, u)$. So, we can equivalently show the following relation holds:
\begin{equation}\label{eq::h_ineq}
 \displaystyle\frac{h''(\rho, u, z)}{h'(\rho, u, z)} \geq  \displaystyle\frac{h''(\rho, u, 0)}{h'(\rho, u, 0)}.
\end{equation}

We first apply the transformations 
\begin{equation*}
u = \tanh(k), \hspace{5mm} z = \tanh(w)  
\end{equation*}
where $k, w\in[0, \infty)$. For shorthand notation, let $h(\rho, \tanh(k), \tanh(w)) \triangleq \tilde{h}(\rho, k, w)$. Using these, we obtain
\begin{equation*}
\tilde{h}(\rho, k, w) = \displaystyle\frac{\cosh(\frac{1}{1+\rho}(k+w))^{1+\rho} + \cosh(\frac{1}{1+\rho}(k-w))^{1+\rho}}{2\cosh(k)\cosh(w)}.
\end{equation*}
Then,
\begin{multline}\label{eq::R_factor}
\displaystyle\frac{\displaystyle\frac{\partial  h(\rho, k, w)}{\partial k^2}}{\displaystyle\frac{\partial h(\rho, k, w)}{\partial k}} = -2 \tanh(k) + \frac{\rho}{1+\rho}\cosh(k) \times \\
 \left[\displaystyle\frac{\cosh(\frac{1}{1+\rho}(k+w))^{\rho-1} + \cosh(\frac{1}{1+\rho}(k-w))^{\rho-1}}{\cosh(\frac{1}{1+\rho}(k+w))^{\rho}\sinh(\frac{\rho}{1+\rho}k-\frac{1}{1+\rho}w) 
+ \cosh(\frac{1}{1+\rho}(k-w))^{\rho}\sinh(\frac{\rho}{1+\rho}k+\frac{1}{1+\rho}w)}\right].
\end{multline}
We note that the additive term $-2 \tanh(k)$, and the non-negative multiplicative factor $\frac{\rho}{1+\rho}\cosh(k)$ do not depend on $w$. Hence, we only need to show the term inside the parenthesis is smallest 
when evaluated at $w = 0$. For this purpose, we define the transformations
\begin{equation*}
 a = \frac{k + w}{1+\rho}, \hspace{5mm} b = \frac{k - w}{1+\rho} 
\end{equation*}
such that $k = (1+\rho)\displaystyle\frac{a+b}{2}$, and $w = (1+\rho)\displaystyle\frac{a-b}{2}$. The condition $k, w\geq 0$ is equivalent to $a \geq |b|$. 
Using these transformations, the reciprocal of the term inside parenthesis in equation \eqref{eq::R_factor} becomes  
\begin{equation*}
 R(\rho, a, b) = \frac{\cosh(b)^{1-\rho} \cosh(a) \sinh(\frac{a+b}{2}\rho - \frac{a-b}{2}) + \cosh(a)^{1-\rho} \cosh(b) \sinh(\frac{a+b}{2}\rho + \frac{a-b}{2})}{\cosh(a)^{1-\rho} + \cosh(b)^{1-\rho}}.  
\end{equation*}
Therefore, the inequality given in \eqref{eq::h_ineq} will hold iff
\begin{equation}\label{eq::R_ineq}
 R(\rho, a, b) \leq R(\rho, \frac{a+b}{2}, \frac{a+b}{2}) = \cosh(\frac{a+b}{2})\sinh(\frac{a+b}{2}\rho).
\end{equation}

We define
\begin{multline*}
 f(\rho, a, b) \triangleq \cosh(\frac{a+b}{2}) \sinh(\rho\frac{a+b}{2}) \left[\cosh(a)^{1-\rho} + \cosh(b)^{1-\rho}\right] \\
- \cosh(a)^{1-\rho}\cosh(b) \sinh(\rho \frac{a+b}{2} + \frac{a-b}{2}) - \cosh(b)^{1-\rho} \cosh(a) \sinh(\rho\frac{a+b}{2} - \frac{a-b}{2}).
\end{multline*}
We note that $f(\rho, a, b)\geq 0$ is equivalent to the inequality \eqref{eq::R_ineq}, which in turn is equivalent to the inequality \eqref{eq::h_ineq}.\\

After simplifications, the function reduces to the following form:
\begin{equation*}
  f(\rho, a, b) = \sinh(\frac{a-b}{2})J(\rho, a, b)
\end{equation*}
where 
\begin{equation*}
J(\rho, a, b) \triangleq \cosh(b)^{1-\rho}\cosh(a - \rho\frac{a+b}{2}) - \cosh(a)^{1-\rho}\cosh(b - \rho\frac{a+b}{2}).
\end{equation*}
Since for $a\geq|b|$, we have
\begin{equation*}
\sinh(\frac{a-b}{2})\geq 0
\end{equation*}
we only need to show that $J(\rho, a, b) \geq 0$. \\

We introduce the variables $k'$, and $w'$ using $a = k'+w'$, and $b = k'-w'$ where $k', w'\in[0, \infty)$. Then, we get
\begin{equation*}
J(\rho, k'+w', k'-w') = \cosh(k' - w')^{1 - \rho} \cosh(k' - \rho k' + w') - \cosh(k' - \rho k' - w') \cosh(k' + w')^{1 - \rho}.
\end{equation*}
We note that $J(\rho, k'+w', k'-w')\Bigl\lvert_{k'=0} = 0$. Moreover, $J(\rho, k'+w', k'-w')$ is increasing in the variable $k'$:
taking the first derivative with respect to $k'$, we get
\begin{equation*}
 \displaystyle\frac{\partial}{\partial k'} J(\rho, k'+w', k'-w') = (1-\rho)\left[\cosh(k'-w')^{-\rho} - \cosh(k'+w')^{-\rho}\right] \sinh((2-\rho)k') \geq 0
\end{equation*}
where the positivity follows from the fact that $|k' - w'| \leq |k' - w'|$, thus $\cosh(k'-w') \leq \cosh(k'+w')$, and $\cosh(k'-w')^{-\rho} \geq \cosh(k'+w')^{-\rho}$, and from the fact that 
$\sinh(x)\geq 0$ holds for $\forall x \geq 0$. \\

As a result, $J(\rho, k'+w', k'-w') \geq 0$ as required, and we have shown that the inequality given in \eqref{eq::h_ineq} holds. This concludes the proof.  
\end{proof}

\end{document}